\documentclass[11pt]{article}
\usepackage[margin=1in]{geometry}
\usepackage{amsmath,amsthm,amsfonts,amssymb,amscd,bm,mathrsfs}
\usepackage{graphicx}
\usepackage[colorlinks=true, urlcolor=blue, linkcolor=blue, citecolor=magenta]{hyperref}
\usepackage{multirow}
\usepackage{enumitem}
\usepackage[T1]{fontenc}
\usepackage{aliascnt}

\allowdisplaybreaks

\usepackage[capitalise,nameinlink]{cleveref}
\Crefname{lemma}{Lemma}{Lemmas}
\Crefname{fact}{Fact}{Facts}
\Crefname{theorem}{Theorem}{Theorems}
\Crefname{corollary}{Corollary}{Corollaries}
\Crefname{claim}{Claim}{Claims}
\Crefname{example}{Example}{Examples}
\Crefname{problem}{Problem}{Problems}
\Crefname{definition}{Definition}{Definitions}
\Crefname{notation}{Notation}{Notations}
\Crefname{assumption}{Assumption}{Assumptions}
\Crefname{subsection}{Subsection}{Subsections}
\Crefname{section}{Section}{Sections}
\Crefformat{equation}{(#2#1#3)}
\providecommand{\footref}[1]{\textsuperscript{\ref{#1}}}

\newtheorem{theorem}{Theorem}[section]
\newtheorem*{theorem*}{Theorem}

\newcommand{\newaliastheorem}[2]{%
  \newaliascnt{#1}{theorem}%
  \newtheorem{#1}[#1]{#2}%
  \aliascntresetthe{#1}%
  \crefname{#1}{#2}{#2s}%
  \Crefname{#1}{#2}{#2s}%
}

\newaliastheorem{proposition}{Proposition}
\newtheorem*{proposition*}{Proposition}

\newaliastheorem{lemma}{Lemma}
\newtheorem*{lemma*}{Lemma}

\newaliastheorem{corollary}{Corollary}
\newtheorem*{corollary*}{Corollary}

\newaliastheorem{fact}{Fact}
\newtheorem*{fact*}{Fact}

\newaliastheorem{exercise}{Exercise}
\newtheorem*{exercise*}{Exercise}

\newaliastheorem{hypothesis}{Hypothesis}
\newtheorem*{hypothesis*}{Hypothesis}

\newaliastheorem{conjecture}{Conjecture}
\newtheorem*{conjecture*}{Conjecture}

\newaliastheorem{subclaim}{Claim}
\newaliastheorem{claim}{Claim}
\newtheorem*{claim*}{Claim}

\theoremstyle{definition}
\newaliastheorem{definition}{Definition}
\newaliastheorem{notation}{Notation}
\newaliastheorem{construction}{Construction}
\newaliastheorem{example}{Example}
\newaliastheorem{question}{Question}
\newaliastheorem{openquestion}{Open Question}
\newaliastheorem{problem}{Problem}
\newaliastheorem{protocol}{Protocol}
\newaliastheorem{assumption}{Assumption}

\newaliastheorem{exercise-easy}{Exercise}
\newaliastheorem{exercise-med}{Exercise}
\newaliastheorem{exercise-hard}{Exercise$^\star$}
\newaliastheorem{remark}{Remark}
\newtheorem*{remark*}{Remark}

\newaliastheorem{observation}{Observation}
\newtheorem*{observation*}{Observation}

\newcommand{\wtwo}{\mathsf{W}[2]}
\newcommand{\wone}{\mathsf{W}[1]}
\newcommand{\wpp}{\mathsf{W[P]}}
\newcommand{\fpt}{\mathsf{FPT}}
\newcommand{\hypoeth}{\mathsf{ETH}}
\newcommand{\seth}{\mathsf{SETH}}
\newcommand{\gapeth}{\textsf{Gap-}\hypoeth}

\newcommand{\maxcover}{\mathsf{MaxCover}}
\newcommand{\minlabel}{\mathsf{MinLabel}}

\newcommand{\poly}{\mathsf{poly}}
\newcommand{\polylog}{\mathsf{polylog}}

\makeatletter
\newcommand{\card}{\@ifstar{\card@star}{\card@nostar}}
\newcommand{\card@star}[1]{\left\lvert #1 \right\rvert}
\newcommand{\card@nostar}[1]{\lvert #1 \rvert}
\makeatother

\newcommand{\mmsatwo}{$k$-MMSA$_2$}
\newcommand{\mmsathree}{$k$-MMSA$_3$}
\newcommand{\mmsafour}{$k$-MMSA$_4$}

\title{On the Approximability of Parameterized Minimum Monotone Satisfying Assignment}
\author{
Venkatesan Guruswami\thanks{Simons Institute for the Theory of Computing, and Departments of EECS \& Mathematics, UC Berkeley. Email: \texttt{venkatg@berkeley.edu}. Research supported in part by NSF grant CCF-2211972, a DARPA grant under Contract No. HR0011262E031, and a Simons Investigator award.}
\and
Bingkai Lin\thanks{State Key Laboratory of Novel Software Technology, Nanjing University. Email: \texttt{lin@nju.edu.cn}.}
\and
Xuandi Ren\thanks{Department of EECS, UC Berkeley. Email: \texttt{xuandi\_ren@berkeley.edu}. Supported in part by NSF grant CCF-2228287 and V.G.'s Simons Investigator award.}
\and
Xin Zheng\thanks{State Key Laboratory of Novel Software Technology, Nanjing University. Email: \texttt{xinzheng@smail.nju.edu.cn}.}
}

\date{}

\begin{document}

\maketitle
\thispagestyle{empty}

\begin{abstract}
    The parameterized Minimum Monotone Satisfying Assignment ($k$-MMSA) problem asks whether a monotone Boolean circuit admits a satisfying assignment of Hamming weight at most $k$. The MMSA hierarchy is defined by allowing a bounded number of alternations between AND and OR gates in the circuit. While the polynomial-time approximability of the MMSA hierarchy has been studied extensively, much less is known in the parameterized setting. In particular, $k$-MMSA$_2$ is the well-known $k$-SetCover problem, whose parameterized inapproximability lies in the $\polylog(n)$ regime. In contrast, $k$-MMSA$_4$  captures $k$-MinLabel, for which known lower bounds give $\poly(n)$ inapproximability. Sandwiched by $k$-MMSA$_2$ and $k$-MMSA$_4$, the inapproximability of $k$-MMSA$_3$ remained comparatively unexplored. 
    
    In this paper, we give an FPT-time $O(2^k \log n)$-approximation algorithm for $k$-MMSA$_3$, suggesting that in the fixed-parameter regime, the third level of MMSA remains surprisingly close to the second level. Complementing this algorithm, we also give an FPT-time gap-preserving reduction from $k$-MMSA$_3$ to $k$-MMSA$_2$. Thus, stronger inapproximability for $k$-MMSA$_3$ would imply new hardness for $k$-MMSA$_2$, potentially offering a route around the current barriers for the latter problem.

    Revisiting Marx's reduction from $k$-MMSA$_t$ to gap $k$-MMSA$_{t+2}$, we also show that $k$-MMSA$_4$ admits no $n^{o(1)}$-factor FPT approximation unless W[2]=FPT, and no $n^{O(1/k)}$-factor approximation running in $n^{o(k)}$ time under ETH. These results separate the parameterized approximability behavior of the third and fourth levels and clarify where stronger inapproximability enters the $k$-MMSA hierarchy.
\end{abstract}

\clearpage

\section{Introduction}
Given a Boolean circuit $C$ consisting of only $\land$ and $\lor$ gates on $n$ input bits, the Minimum Monotone Satisfying Assignment (MMSA) problem asks to find a satisfying assignment of $C$ with minimum Hamming weight. 

MMSA was first introduced by Alekhnovich, Buss, Moran and Pitassi \cite{ABRP01}.  In a later work, Dinur and Safra \cite{DS04} defined a hierarchy of MMSA problems by restricting the number of alternations of the monotone circuit\footnote{Strictly speaking, \cite{DS04} studied  monotone \textit{formulae} rather than general \textit{circuits}. For fixed depth $t$, the two are interconvertible with polynomial blow-up. The distinction matters only in the unbounded-depth setting, e.g., in our later discussion on $\wpp$/$\mathsf{W[SAT]}$-completeness of unbounded-alternation $k$-MMSA.}. Specifically, MMSA$_t$ refers to the problem where the alternation starts with an $\land$, and has at most $t$ levels. This hierarchy turns out to be equivalent to an AND/OR scheduling hierarchy suggested by Goldwasser and Motwani \cite{GM97}. 

Looking into this hierarchy, MMSA$_1$ can be solved easily by simply outputting the width of the $\land$ gate. MMSA$_2$ is equivalent to the classical Set Cover  problem, by identifying each set with a variable and each element to cover with an inner disjunction. Thus, MMSA$_2$ inherits the $O(\log n)$-approximation \cite{joh74,Lov75,cha79} and $(1-o(1)) \ln n$ NP-inapproximability \cite{dinste14} of Set Cover. In \cite{DS04}, Dinur and Safra showed a $2^{\log ^{1-o(1)} n}$ inapproximability for MMSA$_3$ via a reduction from PCP, and placed the Label-Cover problem\footnote{This is the minimization Label-Cover problem in the Arora--Lund sense
\cite{AL97}, rather than the Projection-game formulation, where the value is the maximum
fraction of satisfied constraints.} between MMSA$_3$ and MMSA$_4$. Recently, Chlamt\'{a}\v{c}, Makarychev, and Vakilian~\cite{ChlamtacMV23}
revisited the polynomial-time approximability of this hierarchy. They gave an $\widetilde O(N^{1/3})$-approximation for MMSA$_3$ and
MMSA$_4$, and more generally an
	$\widetilde O(N^{1-\delta_t})$-approximation for MMSA$_t$,
	where $N$ is the total number of gates and variables, and $\delta_t=\frac13\cdot 2^{3-\lceil t/2\rceil}$.

In a separate line of work, Umans \cite{Umans99} studied the \textit{semantic} variant of Minimum Monotone Satisfying Assignment problem. In his formulation, the input is a circuit that is promised to accept a monotone set of strings. In contrast, \cite{ABRP01, DS04} studies the syntactic variant where the circuit itself is required to be monotone. Umans showed the NP-hardness of an $n^{1/5-\varepsilon}$-approximation for the semantic problem.  

In the parameterized world, there is even more connection between MMSA and other important problems. Parameterized by the target Hamming weight $k$, the unbounded-alternation circuit version of $k$-MMSA is $\wpp$-complete, whereas the corresponding formula version is $\mathsf{W[SAT]}$-complete; imposing constant-depth and weft-$t$ restrictions on weighted circuit satisfiability yields the usual $\mathsf{W}[t]$-hierarchy \cite{FlumGrohe06,downey2013fundamentals}. 
As for the bounded-depth variants, \mmsatwo{} is equivalent to {\sc $k$-Setcover}, the canonical $\wtwo$-complete problem \cite{DowneyF95,FlumGrohe06,downey2013fundamentals}. Inapproximability for {\sc $k$-Setcover} was known under different complexity assumptions: there are no $n^{o(k)}$ time (respectively, $n^{k-\varepsilon}$ time) $(\frac{\log n}{\log \log n})^{1/k}$-approximation under $\hypoeth$ (respectively, $\seth$) \cite{Lin19}; no FPT time $(\log n)^{o(1)}$-approximation under $\wone\neq\fpt$ \cite{KLM19,Lin19}\footnote{The factor in the statement of \cite{KLM19} is $(\log n)^{1/\poly(k)}$, but by a simple padding trick, we can actually get any FPT-time $(\log n)^{o(1)}$-factor inapproximability. See \Cref{lem:mmsa_padding} for a formal discussion.}; and no FPT time constant approximation under $\wtwo\neq\fpt$ \cite{LRSW23a}. 

Despite these substantial results, the parameterized inaprpoximability of {\sc $k$-Setcover} is not yet fully understood. Several central questions remain open and have been raised repeatedly in the literature \cite{feldmann2020survey, KLM19, Lin19, KL21, LRSW23a}. Can the FPT-time inapproximability factor be strengthened to match the $O(\log n)$-approximation by the greedy algorithm? Can the total-FPT-inapproximability\footnote{By total-FPT-inapproximability, we mean ruling out FPT-time $f(k)$-approximation algorithms for every computable function $f$.} be based  on the minimal assumption $\wtwo\neq\fpt$? Under $\hypoeth$, or even $\gapeth$, can one go beyond the current $(\log n)^{1/k}$-scale barrier for $n^{o(k)}$-time inapproximability?
These questions motivate the study of $k$-MMSA$_3$, the closest natural extension of $k$-MMSA$_2$ in the $k$-MMSA hierarchy. In particular, we ask whether the additional expressive power of $k$-MMSA$_3$ enables reductions that bypass the known barriers for {\sc $k$-SetCover}, and whether the techniques developed for  $k$-MMSA$_3$ can in turn shed light on {\sc $k$-SetCover} itself.

We remark that Dinur and Safra's approach \cite{DS04} for non-parameterized MMSA$_3$ relies on building one Boolean variable $L[x,a]$ for each (variable, value) pair in the PCP construction. In the yes case, selecting exactly one value per PCP variable yields a satisfying assignment of weight $n$, where $n$ is the number of variables. 
However, this approach does not extend to the parameterized setting: to keep the solution weight bounded by $k$, one would need to start from a \textit{parameterized} CSP with only $k$ variables, but for such CSPs, even just greedily satisfying the constraints adjacent to the maximum degree variable gives a $k^{O(1)}$ factor approximation\footnote{We may assume without loss of generality that the parameterized CSP has arity 2. Arity-$q$ CSPs reduce to arity 2 by the standard label-coverization trick, with only an $O(q)$-factor loss in the gap.}, which precludes deriving any inapproximability factor that grows with $n$.

With an extra level of alternation, \mmsafour{} 
possesses a stronger inapproximability. In his paper studying parameterized monotone and antimonotone circuit satisfiability, Marx \cite{Marx13} gave a gap-creating reduction from any $k$-MMSA$_t$ to $k$-MMSA$_{t+2}$, and thus ruled out any FPT-time $f(k)$-approximation for \mmsafour{} under $\wtwo\neq\fpt$. We remark that \mmsafour{} also captures the {\sc $k$-MinLabel} problem defined in \cite{CCK+17}. It was known in \cite{CCK+17} that, under $\gapeth$, any $f(k)$-approximation of {\sc $k$-MinLabel} requires $n^{\Omega(k)}$ time. It was also implicit in \cite{KLM19, KL21} that, {\sc $k$-MinLabel} is hard to approximate within an $n^{o(1)}$-factor\footnote{\label{fn:padding}Requires padding, see \Cref{lem:mmsa_padding}.} in FPT time under $\wone\neq\fpt$, or an $n^{1/\text{poly}(k)}$-factor in $n^{o(k)}$ time under $\hypoeth$, via reductions from {\sc $k$-MaxCover}.

\paragraph*{Our Contributions.}
In this paper, we give an FPT time $O(2^k \log n)$-approximation algorithm for \mmsathree{} (\Cref{thm:logn_alg_mmsathree}), which shows, in contrast to the non-parameterized world, that approximating \mmsathree{} is not significantly harder than approximating \mmsatwo{}. On the hardness side, we give a reduction from $(k,h)$-gap \mmsathree{} to $(2^k, h/k)$-gap \mmsatwo{} (\Cref{thm:mmsa3_to_mmsa2_reduction}).
This connection has two implications. First, it helps
explain why stronger inapproximability for $k$-MMSA$_3$ has not
appeared separately from \mmsatwo{}. Second, via this reduction, any sufficiently strong lower bound for $k$-MMSA$_3$ would yield new FPT inapproximability results for {\sc $k$-SetCover}. In particular, a $(\log n)^{\Omega(1)}$-factor FPT inapproximability for $k$-MMSA$_3$ under $\wone\neq\fpt$ would imply the same for
{\sc $k$-SetCover}. Similarly, under $\wtwo\neq\fpt$, an $\omega(2^k)$ FPT inapproximability for
$k$-MMSA$_3$ would translate into a
	superconstant FPT inapproximability for {\sc $k$-SetCover}, resolving long standing open problems.

We also revisit Marx’s reduction \cite{Marx13} and show that it yields stronger consequences than previously recognized. By carefully picking the parameters, we show \mmsafour{} is hard to approximate within an $n^{o(1)}$-factor\footref{fn:padding} in FPT time under $\wtwo\neq\fpt$, and hard to approximate within an $n^{O(1/k)}$-factor in $n^{o(k)}$ time under $\hypoeth$, improving the previous $f(k)$-factor \cite{Marx13} and $n^{1/\text{poly}(k)}$-factor \cite{KLM19,KL21}, respectively.

We summarize the (in)approximability results for parameterized MMSA hierarchy in \Cref{tab:k-mmsa-approx-hardness}.

\begin{table}[!htbp]
\centering
\renewcommand{\arraystretch}{1.25}
\setlength{\tabcolsep}{2pt}
\resizebox{\textwidth}{!}{%
\begin{tabular}{c|c|c|c|c|c|c|c}
\multirow{2}{*}{\textbf{Problems}}
& \multicolumn{3}{|c|}{\textbf{Approximability}}
& \multicolumn{4}{c}{\textbf{Inapproximability}} \\
\cline{2-8}
& \textbf{Ratio}
& \textbf{Runtime}
& \textbf{Reference}
& \textbf{Ratio}
& \textbf{Runtime}
& \textbf{Assumption}
& \textbf{Reference} \\
\hline
\multirow{4}{*}{\mmsatwo{}}
& \multirow{4}{*}{$O(\log n)$}
& \multirow{4}{*}{$\poly(n)$}
& \multirow{4}{*}{\shortstack[c]{\cite{joh74,Lov75}\\\cite{cha79}}}
& $\left(\dfrac{\log n}{\log\log n}\right)^{1/k}$
& \multirow{2}{*}{$n^{o(k)}$}
& \multirow{2}{*}{$\hypoeth$}
& \cite{Lin19} \\
\cline{5-5}\cline{8-8}
& & &
& $(\log n)^{1/\poly(k)}$
&
&
& \cite{KLM19} \\
\cline{5-8}
& & &
& $(\log n)^{o(1)}$
& \multirow{2}{*}{$\fpt$}
& $\wone \neq \fpt$
& \cite{KLM19,Lin19} \\
\cline{5-5}\cline{7-8}
& & &
& any $O(1)$
&
& $\wtwo \neq \fpt$
& \cite{LRSW23a} \\

\hline
	\mmsathree{}
& $O(2^k\log n)$
& $\fpt$
& \Cref{thm:logn_alg_mmsathree}
& \multicolumn{4}{|c}{Inherited from above} \\
\hline
\multirow{5}{*}{\mmsafour{}}
& \multirow{5}{*}{$n$}
& \multirow{5}{*}{$\poly(n)$}
& \multirow{5}{*}{Trivial}
& $n^{1/\poly(k)}$
& \multirow{2}{*}{$n^{o(k)}$}
& \multirow{2}{*}{$\hypoeth$}
& \cite{KLM19,KL21} \\
\cline{5-5}\cline{8-8}
& & &
& $n^{O(1/k)}$
&
&
& \Cref{cor:mmsafour_eth} \\
\cline{5-8}
& & &
& $n^{o(1)}$
& \multirow{3}{*}{$\fpt$}
& $\wone \neq \fpt$
& \cite{KLM19,KL21} \\
\cline{5-5}\cline{7-8}
& & &
& any $f(k)$
&
& \multirow{2}{*}{$\wtwo \neq \fpt$}
& \cite{Marx13} \\
\cline{5-5}\cline{8-8}
& & &
& $n^{o(1)}$
&
&
& \Cref{cor:mmsafour_wtwo} \\
\hline
\end{tabular}
}
\caption{Approximability and Inapproximability Results for Parameterized $k$-MMSA.}
\label{tab:k-mmsa-approx-hardness}
\end{table}

\section{Preliminaries}

\subsection{Perfect Hash Families}

\begin{definition}[$(n,m,\Sigma,h)$-perfect hash family \cite{FK84}]
    An $(n,m,\Sigma,h)$-perfect hash family is a family of  $m$ functions $\mathcal F=\{f_i:[n] \to \Sigma \mid i \in [m]\}$, such that for every subset $T\subseteq [n]$ where $|T| \le h$, there is a function $f \in \mathcal F$ such that
    $$\text{$\forall x,y \in T$ with $x\neq y$},\  f(x) \neq f(y).$$
\end{definition}

In other words, $\mathcal F$ is an $(n,m,\Sigma,h)$-perfect hash family if for any $h$ distinct elements in $[n]$, there is a function $f \in \mathcal F$ which maps them to $h$ distinct values.

When $|\Sigma| \ge h^2$, an $(n,m,\Sigma,h)$-perfect hash family can be constructed efficiently and deterministically:

\begin{lemma}[\cite{AYZ95}]
\label{lem:perfect_hash_family}
    An $(n,m,\Sigma,h)$-perfect hash family with $|\Sigma|=h^2, m=h^{O(1)} \log n$ can be constructed deterministically in time $h^{O(1)}\cdot n \log n$.
\end{lemma}

\subsection{Parameterized Problems}

We begin with a definition of the parameterized version of the MMSA problem.

\begin{definition}[$k$-MMSA$_t$]
For any $t,k\in\mathbb{N}^+$, $k$-MMSA$_t$ asks whether there is a satisfying assignment with Hamming weight at most $k$, for a monotone circuit $C$ of the following form.
\[
C(x_1,\ldots,x_n)=\left\{\begin{aligned}
        &\bigwedge_{i_1=1}^n \bigvee_{i_2=1}^n \cdots \bigwedge_{i_t=1}^n x_{w_{i_1,\ldots,i_t}}, & & t \text{ is odd},\\
        &\bigwedge_{i_1=1}^n \bigvee_{i_2=1}^n \cdots \bigwedge_{i_{t-1}=1}^n\bigvee_{i_t=1}^n x_{w_{i_1,\ldots,i_t}}, & & t \text{ is even},
    \end{aligned}\right.
\]
    where for every $i_1,\ldots,i_t \in [n]$, $w_{i_1,\ldots,i_t}$ is an index in $[n]$\footnote{We generally allow the width of each gate to be polynomial in the number of input bits, but for simplicity, we can pad unused bits and thus use $n$ for both of them.}.

    We write $\operatorname{SAT}(C)$ to denote the set of satisfying assignments of $C$, and write $\operatorname{OPT}(C)$ to denote the minimum weight of any satisfying assignment.

    For any $h>k$, we write $(k,h)$-gap MMSA$_t$ for the task of distinguishing between $\operatorname{OPT}(C)\le k$ and $\operatorname{OPT}(C)>h$,
    or equivalently, finding a satisfying assignment of $C$ with weight at most $h$ given the promise that $\operatorname{OPT}(C)\le k$.
\end{definition}
For the special case of \mmsathree{}, it will be convenient to replace the ordered list of literals feeding into each bottom conjunction by the set of input indices appearing in that gate. Specifically, a \mmsathree{} instance can be written as
    $$
    C(x)=\bigwedge_{i=1}^{n}\bigvee_{j=1}^{n}\bigwedge_{t\in N_{i,j}}x_t,
    $$
    where $N_{i,j}\subseteq[n]$ for every $i,j\in[n]$.
    
The following lemma shows that given a $k$-MMSA$_3$ instance, one can prune any bottom conjunction term with $|N_{i,j}|>k$ without hurting the gap.

\begin{lemma}
\label{lem:prune}
    There is a polynomial-time deterministic algorithm which, given as input a \mmsathree{} instance
    $$
    C(x)=\bigwedge_{i=1}^{n}\bigvee_{j=1}^{n}\bigwedge_{t\in N_{i,j}}x_t,
    $$
    outputs a \mmsathree{} instance
    \[
    C'(x)=\bigwedge_{i=1}^{n}\bigvee_{j=1}^n\bigwedge_{t\in N'_{i,j}}x_t,
    \]
    where for every $i\in[n]$ and $j \in S_i$, $|N_{i,j}| \le k$. Furthermore,
    \begin{itemize}
        \item If $\operatorname{OPT}(C) \le k$, then $\operatorname{OPT}(C') \le k$.
        \item $\operatorname{SAT}(C')\subseteq \operatorname{SAT}(C)$. In particular, for every $h>k$, if $\operatorname{OPT}(C)>h$, then $\operatorname{OPT}(C')>h$.
    \end{itemize}
\end{lemma}
\begin{proof}
    $C'$ is obtained from $C$ by simply removing every bottom conjunction with $|N_{i,j}|>k$. That is, let $S_i:=\{j\in[n]\mid |N_{i,j}|\le k\}$, and let
    $$
        C'(x):=\bigwedge_{i=1}^{n}\bigvee_{j\in S_i}\bigwedge_{t\in N_{i,j}}x_t.
    $$
    To pad the size of the second layer to $n$, we simply duplicate its inputs.

    Obviously, $C(x)\ge C'(x)$, hence $\operatorname{SAT}(C')\subseteq \operatorname{SAT}(C)$.
    On the other hand, suppose $\operatorname{OPT}(C) \le k$, and let $x$ be a satisfying assignment to $C$ with $\|x\|_0\le k$.
    Since $|N_{i,j}|>k$ for every $i\in[n]$ and $j\notin S_i$, we have $\bigwedge_{t\in N_{i,j}}x_t=0$, thus
    $$\begin{aligned}
        C'(x)&=\bigwedge_{i=1}^{n}\left(\bigg(\bigvee_{j\in S_i}\bigwedge_{t\in N_{i,j}}x_t\bigg)\vee\bigg(\bigvee_{j\notin S_i}\bigwedge_{t\in N_{i,j}}x_t\bigg)\right) \\
        &=\bigwedge_{i=1}^{n}\bigvee_{j=1}^{n}\bigwedge_{t\in N_{i,j}}x_t=C(x).
    \end{aligned}$$
    Hence $x$ is also a satisfying assignment to $C'$, which implies that $\operatorname{OPT}(C') \le k$.
\end{proof}

We next introduce two closely related parameterized problems: $k$-MaxCover and $k$-MinLabel.

\begin{definition}[$k$-MaxCover and $k$-MinLabel]
\label{def:maxcover_minlabel}
    $k$-MaxCover and $k$-MinLabel are both defined on a partitioned bipartite graph $G=(A,B,E)$ where
    $A=A_1\dot\cup \ldots\dot\cup A_k, B= B_1\dot\cup \ldots\dot\cup B_m$ and $m \le n$, with different optimizing objectives.

    We say a set $X \subseteq A$ (possibly contain multiple vertices from the same part) covers a right part $B_j$ if there exists some $b \in B_j$, such that for every $i \in [k]$, there exists a vertex $a_i \in A_i \cap X$ which is connected to $b$.
    
    The goal of $k$-MaxCover is to pick one vertex $a_i$ from each left part $A_i$, to cover the maximum fraction of right parts $B_j$, while the goal of $k$-MinLabel is  to pick minimum-sized subset $X \subseteq A$, to cover every right part $B_j$. We write the two values as $\maxcover(G)$ and $\minlabel(G)$, respectively. Specifically,
    \[
\begin{gathered}
    \maxcover(G)
    =
    \frac{1}{m}\left(
    \max_{a_1 \in A_1,\ldots,a_k \in A_k}
    \card*{
    \left\{
        j \in [m]
        \;\big|\;
        \exists b \in B_j
        \text{ such that }
        \forall i \in [k],\ (a_i,b) \in E
    \right\}
    }\right),
    \\[2mm]
    \minlabel(G)
    =
    \min_{X \subseteq A
        , X \text{ covers every } B_j,\ j \in [m]}
        \card*{X}.
\end{gathered}
\]
\end{definition}

There is a canonical reduction from $(1,\varepsilon)$-gap $k$-MaxCover to $(1,\varepsilon^{-1/k})$-gap $k$-MinLabel, as formalized in \cite{CCK+17,KLM19}. We omit the proof here and refer the reader to \cite[Proposition A.1]{KLM19}.

\begin{lemma}[\cite{CCK+17,KLM19}]
\label{lem:maxcover2minlabel}
    Let $G=(A,B,E)$ be a $k$-MaxCover/$k$-MinLabel instance with $A=A_1\dot\cup\cdots\dot\cup A_k$ and
    $B=B_1\dot\cup\cdots\dot\cup B_m$. 
    \begin{itemize}
        \item If $\maxcover(\Gamma)=1$, then $\minlabel(\Gamma)=k$.
        \item If $\maxcover(\Gamma) \le \varepsilon$ for some $\varepsilon>0$, then $\minlabel(\Gamma)\ge (1/\varepsilon)^{1/k}\cdot k$.
    \end{itemize}
\end{lemma}
We remark that $k$-MinLabel can be written as a restricted special case of
MMSA$_4$. Therefore, the inapproxiambility of $k$-MinLabel applies to $k$-MMSA$_4$ as well. This observation is formalized in the following lemma.

\begin{lemma}
\label{lem:minlabel_mmsafour}
    Let $G=(A,B,E)$ be a $k$-MinLabel instance with
    $A=A_1\dot\cup\cdots\dot\cup A_k$ and
    $B=B_1\dot\cup\cdots\dot\cup B_m$. There is a polynomial-time construction of
    a $k$-MMSA$_4$ circuit $C_G$ such that
    \[
        \operatorname{OPT}(C_G)=\minlabel(G).
    \]
\end{lemma}
\begin{proof}
For every left vertex $a\in A$, introduce one Boolean variable $x_a$. For a
right vertex $b\in B$, write
\[
    N_G(b)=\{a\in A\mid (a,b)\in E\}.
\]
Define the monotone circuit
\[
    C_G(x)
    =
    \bigwedge_{j\in[m]}
    \bigvee_{b\in B_j}
    \bigwedge_{i\in[k]}
    \bigvee_{a\in A_i\cap N_G(b)}
    x_a .
\]
After padding dummy gates/variables to make the width of each gate equal to the number of variables, $C_G(x)$ becomes a $k$-MMSA$_4$ circuit and can be constructed in polynomial time. It's also easy to check the one-to-one correspondence between a subset $X \subseteq A$ covering every $B_j$ in $G$ and a satisfying assignment of $C_G$. \qedhere
\end{proof}

We next present a canonical padding lemma which boosts any FPT-time inapproximability for $k$-MMSA with factor $f(n)^{1/g(k)}$ to $f(n)^{o(1)}$.

\begin{lemma}[Boosting the inapproximability of $k$-MMSA$_t$ by padding]
\label{lem:mmsa_padding}
Fix $t\in\mathbb N^+$. Let $f:\mathbb N\to\mathbb R^+$ be a computable, nondecreasing, unbounded, and padding-stable function in the following sense: for every constant $a\in\mathbb N$, there is a sufficiently large $n_0(a)$ such that
\[
    f(n+a)\le f(n)^2 \qquad\forall n\ge n_0(a).
\]
Let $g:\mathbb N\to\mathbb R^+$ be a computable,  nondecreasing and unbounded function. Suppose for $k$-MMSA$_t$, no FPT algorithm can distinguish between
\[
    \operatorname{OPT}(C)\le k
    \qquad\text{and}\qquad
    \operatorname{OPT}(C)>k\cdot f(n)^{1/g(k)}.
\]
Then, for every computable function $\alpha:\mathbb N\to\mathbb R^+$ with $\alpha(k)=o(1)$, no FPT algorithm can distinguish between
\[
    \operatorname{OPT}(C)\le k
    \qquad\text{and}\qquad
    \operatorname{OPT}(C)>k\cdot f(n)^{\alpha(k)}.
\]
In other words, $k$-MMSA$_t$ doesn't admit an FPT-time $f(n)^{o(1)}$-approximation.
\end{lemma}

\begin{proof}
Fix $\alpha(k)=o(1)$ and choose $K$ sufficiently large such that $\alpha(K)\le 1/(4g(k))$. Given an MMSA$_t$ circuit $C$ on $n$ variables, add $K-k$ fresh variables which are forced to be one in any satisfying assignment:
\[
    \widehat C(x,y)=C(x)\wedge y_1\wedge\cdots\wedge y_{K-k}.
\]
By representing each new variable as a dummy length-$t$ unary path, the new instance $\widehat C$ is still a MMSA$_t$ instance with new parameter $K$. Furthermore, $\widehat C$ has $N=n+K-k$ many variables, satisfies
\[
    \operatorname{OPT}(\widehat C)=\operatorname{OPT}(C)+K-k,
\]
and can be constructed in FPT time. Suppose for contradiction that there is an FPT-time algorithm which can distinguish between
\[
    \operatorname{OPT}(\widehat C)\le K
    \qquad\text{and}\qquad
    \operatorname{OPT}(\widehat C)>K\cdot f(N)^{\alpha(K)}.
\]
Since for sufficiently large $N$ depending only on $K$, padding-stability gives $f(N)=f(n+K-k)\le f(n)^2$, and $f$ is unbounded, we have
\[
 K\cdot f(N)^{\alpha(K)}
    \le K\cdot f(n)^{2\alpha(K)}
    \le K\cdot f(n)^{1/(2g(k))}
    \le K-k+k\cdot f(n)^{1/g(k)}.
\]
Thus, the same algorithm can also distinguish between
\[
    \operatorname{OPT}(C)\le k
    \qquad\text{and}\qquad
    \operatorname{OPT}(C)>k\cdot f(n)^{1/g(k)}
\]
in FPT time, a contradiction.
\end{proof}

\section[An FPT-time O(2 to k log n)-Approximation for k-MMSA3]{An FPT-time $O(2^k\log n)$-Approximation for \mmsathree{}}
\label{sec:mmsa3_alg}

\begin{theorem}
\label{thm:logn_alg_mmsathree}
    There is a deterministic FPT-time $O(2^k\log n)$-approximation algorithm for \mmsathree{} on $n$ input bits. 
\end{theorem}

\begin{proof}
    Write the input instance as
    $$
        C(x)=\bigwedge_{i=1}^{n}\bigvee_{j=1}^{n}\bigwedge_{t\in N_{i,j}}x_t,
    $$
    where $N_{i,j}\subseteq[n]$ for every $i,j\in[n]$. By \Cref{lem:prune}, we can without loss of generality assume that $|N_{i,j}| \le k$.

    The algorithm proceeds by repeatedly assigning some variables to 1, until the whole instance $C$ is satisfied. Let
    $$S=\{i \in [n] \mid x_i=1\}$$
    be the set of variables that are assigned to 1, the $i$-th clause
    $$\left(\bigvee_{j=1}^n \bigwedge_{t \in N_{i,j}} x_t\right)$$
    is covered if and only if there exists some term $j\in[n]$ such that $N_{i,j} \subseteq S$. Let
    $$
        \operatorname{Cov}(S)=\{i\in[n]\mid \exists j \in [n]\text{ with }N_{i,j}\subseteq S\}.
    $$
    The algorithm starts with $S_0=\emptyset$. In the $r$-th iteration, it takes the set $N_{a,b}$ maximizing the number of newly covered indices
    $$
        \left|\operatorname{Cov}(S_{r-1}\cup N_{a,b})\setminus \operatorname{Cov}(S_{r-1})\right|,
    $$
    and sets $S_{r}=S_{r-1}\cup N_{a,b}$. Given the fact that $|N_{i,j}| \le k$ for every $i,j \in [n]$, it suffices to prove that the algorithm stops after $O(2^k \log n)$ iterations, and thus outputs a satisfying assignment of $C$ with weight at most $O(k \cdot 2^k \log n)$.

    Let $x \in \{0,1\}^n$ be a satisfying assignment of $C$ and let $X=\text{supp}(x)$. Fix the $r$-th iteration, let 
    $$
        U_r=[n]\setminus \operatorname{Cov}(S_r)
    $$
    be the set of currently uncovered indices. For every $i\in U_r$, since $x$ satisfies $C$, there exists some witness $j=j(i)\in [n]$ with $N_{i,j(i)} \subseteq X$. Every chosen witness is a subset of $X$, and $X$ has at most $2^k$ distinct subsets. Therefore, by an averaging argument, some set $W\subseteq X$ appears as $N_{i,j(i)}$ for at least $|U_r|/2^k$ indices $i\in U_r$. The best enumerated set covers at least $|U_r|/2^k$ currently uncovered indices. Consequently,
    $$
        |U_{r+1}|\le \left(1-\frac1{2^k}\right)|U_r|.
    $$
    Iterating this inequality, for $T= \lceil 2^k \ln n\rceil$,
    $$
        |U_T|
        \le n\left(1-\frac1{2^k}\right)^T
        \le n\exp\left(-\frac{T}{2^k}\right)<1.
    $$
    Since $|U_T|$ is an integer, $U_T=\emptyset$. Thus, the algorithm returns a satisfying assignment for $C$ with weight at most $kT = O(k 2^k \log n)$.

    The algorithm is deterministic and runs in FPT time: it performs $O(2^k\log n)$ iterations, and in each iteration it enumerates at most $n^2$ candidate sets and computes their gains in polynomial time. This proves the theorem.
\end{proof}

\newcommand{\mmsa}{\mathrm{MMSA}}

\section[A Reduction from (k,h)-gap MMSA3 to (2 to k,h/k)-gap MMSA2]{A Reduction from $(k,h)$-gap MMSA$_3$ to $(2^k, h/k)$-gap MMSA$_2$}

\begin{theorem}
\label{thm:mmsa3_to_mmsa2_reduction}
    There is an FPT-time reduction, which takes as input a \mmsathree{} instance
    $$
        F(z)=\bigwedge_{i\in[n]}\bigvee_{j\in[n]}\bigwedge_{t\in N_{i,j}}z_t,
    $$
    on $z \in \{0,1\}^n$, outputs a $2^k$-MMSA$_2$ instance $G$ on $x\in\{0,1\}^{n^2}$, such that
    \begin{itemize}
        \item (Completeness) If $\operatorname{OPT}(F) \le k$, then $\operatorname{OPT}(G) \le 2^k$. 
        \item (Soundness) Fix any $h>k$, if $\operatorname{OPT}(F) >h$, then $\operatorname{OPT}(G)>h/k$.
    \end{itemize}
\end{theorem}

The reduction contracts each distinct bottom conjunction of $F$ into one variable of the \mmsatwo{} instance $G$.  Namely, for every distinct set $N_{i,j}$ we create a variable intended to mean that all original variables in $N_{i,j}$ are set to $1$.  Since there are at most $n^2$ bottom conjunctions, this gives at most $n^2$ variables.

For completeness, suppose $F$ has a satisfying assignment $z$ whose support $Z$ has size at most $k$. Every bottom conjunction that is true under $z$ corresponds to a set $N_{i,j}\subseteq Z$ of size $\le k$, and there are at most $2^k$ many such sets. Setting the variables representing these true conjunctions to one gives a satisfying assignment of $G$ of weight at most $2^k$.

For soundness, suppose $G$ has a satisfying assignment $x$ of weight at most $h/k$. Decode $x$ back to an assignment $z$ for $F$ by taking the union of all sets represented by the ones in $x$. Since every represented set has size at most $k$, the decoded assignment has weight at most $h$.

\begin{proof}[Proof of \Cref{thm:mmsa3_to_mmsa2_reduction}]
    Let
    $$
        \mathcal S:=\{N_{i,j} \mid i,j\in[n]\}.
    $$
    By \Cref{lem:prune}, we can without loss of generality assume that every set in $\mathcal S$ has size at most $k$.  Since $|\mathcal S|\le n^2$, fix an injective map $f:\mathcal S\to[n^2]$.  For every $i,j\in[n]$, write
    $$
        s_{i,j}:=f(N_{i,j}).
    $$
    For every $s\in\operatorname{Im}(f)$, let $T_s$ denote the unique set in $\mathcal S$ such that $f(T_s)=s$.

    Define the \mmsatwo{} instance
    $$
        G(x):=\bigwedge_{i\in[n]}\bigvee_{j\in[n]}x_{s_{i,j}}.
    $$
    The instance uses at most $n^2$ variables and has size $n^{O(1)}$.  It can clearly be constructed in polynomial time, and hence in FPT time.

    \paragraph*{Completeness.}
    Suppose there exists $z\in\{0,1\}^n$ such that $F(z)=1$ and $\|z\|_0\le k$.  Let
    $$
        Z:=\operatorname{supp}(z)=\{t\in[n]:z_t=1\}.
    $$
    Define $x\in\{0,1\}^{n^2}$ by
    $$
        x_s:=
        \begin{cases}
            1, & \text{if }s\in\operatorname{Im}(f)\text{ and }T_s\subseteq Z,\\
            0, & \text{otherwise.}
        \end{cases}
    $$
    We first show that $G(x)=1$.  Fix any $i\in[n]$.  Since $F(z)=1$, there exists some $j\in[n]$ such that
    $$
        \bigwedge_{t\in N_{i,j}}z_t=1,
    $$
    equivalently, $N_{i,j}\subseteq Z$.  As $T_{s_{i,j}}=N_{i,j}$, the definition of $x$ gives $x_{s_{i,j}}=1$.  Therefore every outer clause of $G$ is satisfied, and hence $G(x)=1$.

    It remains to bound the weight of $x$.  If $x_s=1$, then $s\in\operatorname{Im}(f)$ and $T_s\subseteq Z$.  Because $f$ is injective, distinct coordinates correspond to distinct subsets of $Z$.  Hence
    $$
        \|x\|_0\le 2^{|Z|}\le 2^k.
    $$
    Thus $\operatorname{OPT}(G)\le 2^k$.

    \paragraph*{Soundness.}
    Fix $h>k$, and suppose for contradiction that there exists $x\in\{0,1\}^{n^2}$ such that $G(x)=1$ and $\|x\|_0\le h/k$.  Let
    $$
        \mathcal A:=\{T_s\mid s\in\operatorname{Im}(f)\text{ and }x_s=1\}
    $$
    be the family of original bottom sets represented by the selected variables of $G$, and let
    $$
        U:=\bigcup_{T\in\mathcal A}T.
    $$
    Define $z\in\{0,1\}^n$ to be the indicator vector of $U$.

    Since every $T\in\mathcal A$ has size at most $k$, and since $|\mathcal A|\le\|x\|_0$, we have
    $$
        \|z\|_0=|U|
        \le \sum_{T\in\mathcal A}|T|
        \le k\cdot |\mathcal A|
        \le k\|x\|_0
        \le h.
    $$

    We now show that $F(z)=1$.  Fix any $i\in[n]$.  Since $G(x)=1$, there exists some $j\in[n]$ such that $x_{s_{i,j}}=1$.  By construction, $T_{s_{i,j}}=N_{i,j}$, so $N_{i,j}\in\mathcal A$ and therefore $N_{i,j}\subseteq U=\operatorname{supp}(z)$.  Hence
    $$
        \bigwedge_{t\in N_{i,j}}z_t=1.
    $$
    This holds for every $i\in[n]$, so $F(z)=1$, contradicting $\operatorname{OPT}(F)>h$.
\end{proof}
\begin{corollary}
\label{cor:gap_mmsa3_to_mmsa2}
    For any computable function $t(n,k)$, if there is an FPT-time $t(n,k)$-approximation for \mmsatwo{}, then there is an FPT-time $(t(n^2,2^k)\cdot 2^k)$-approximation for \mmsathree{}. Equivalently, if there is no FPT-time $t(n,k)$-approximation for \mmsathree{}, then there is no FPT-time $\frac{t(\sqrt n,\log k)}{k}$-approximation for \mmsatwo{}.
\end{corollary}

In particular, a $(\log n)^{\Omega(1)}$-factor FPT inapproximability for \mmsathree{} under $\wone\neq\fpt$, or an $\omega(2^k)$-factor FPT inapproximability for \mmsathree{} under $\wtwo\neq\fpt$, could imply new FPT inapproximability results for \mmsatwo{} under the respective assumptions, improving \cite{KLM19,Lin19,LRSW23a}.

\section[A Reduction from k-MMSA-t to Gap k-MMSA-(t+2)]{A Reduction from $k$-MMSA$_t$ to Gap $k$-MMSA$_{t+2}$} 

In this section, we first summarize known inapproximability results for $k$-MMSA$_4$, then introduce \cite{Marx13}'s reduction from $k$-MMSA$_t$ to $k$-MMSA$_{t+2}$, and finally plug a better set of parameters into it to get improved inapproximability of $k$-MMSA$_4$ under $\wtwo\neq\fpt$ and $\hypoeth$.

\subsection[Inapproximability of k-MMSA4 from Literature]{Inapproximability of \mmsafour{} from Literature}

The following inapproximability on $k$-MaxCover was from \cite{KLM19, KL21}. We remark that \cite{KLM19} also provides the inapproximability under other assumptions such as $\seth$ or the $k$-$\mathsf{SUM}$ Hypothesis. However, we choose to present only the most relevant assumptions here.

\begin{lemma}[Inapproximability of $k$-MaxCover \cite{KLM19, KL21}]\label{lem:kmaxcover_hardness}\leavevmode
\begin{itemize}
    \item Assuming $\wone\neq\fpt$, no FPT  algorithm can approximate $k$-MaxCover within an $n^{1/\poly(k)}$ factor. 
    \item Assuming $\hypoeth$, no $n^{o(k)}$-time algorithm can approximate $k$-MaxCover within an $n^{1/\poly(k)}$ factor. 
\end{itemize}
    
\end{lemma}

Passing \Cref{lem:kmaxcover_hardness} through the reduction chain of \Cref{lem:maxcover2minlabel,lem:minlabel_mmsafour,lem:mmsa_padding}, we have:
\begin{corollary}[Inapproximability of $k$-MMSA$_4$, implicit in \cite{KLM19, KL21}]\label{lem:klmkl_mmsa4}\leavevmode
\begin{itemize}
    \item Assuming $\wone\neq\fpt$, no FPT algorithm can approximate $k$-MMSA$_4$ within an $n^{o(1)}$ factor.
    \item Assuming $\hypoeth$, no $n^{o(k)}$-time algorithm can approximate $k$-MMSA$_4$ within an $n^{1/\poly(k)}$ factor.
\end{itemize}
\end{corollary}

Marx \cite{Marx13} also showed a total-FPT-inapproximability for $k$-MMSA$_4$ under the weaker assumption $\wtwo\neq\fpt$.

\begin{lemma}[\cite{Marx13}]
\label{lem:marx13kmmsa4}
    Assuming $\wtwo\neq\fpt$, for any computable function $f$, no FPT algorithm can approximate $k$-MMSA$_4$ within an $f(k)$ factor.
\end{lemma}

\subsection{Revisiting Marx's Reduction}

\begin{theorem}[Adapted from \cite{Marx13}]
\label{thm:mmsa_reduction_phf}
Let $t\in\mathbb N^+$ and let $h\ge k$. 
Given a $k$-MMSA$_t$ instance $C$ on $n$ variables and an
$(n,m,\Sigma,h)$-perfect hash family
$\mathcal F=\{f_1,\ldots,f_m\}$ with $|\Sigma|\ge k$, there is a reduction
that outputs a $k$-MMSA$_{t+2}$ instance $C'$ on the same $n$ variables such that:
\begin{itemize}
    \item the size and running time are at most
    $\left(m\cdot \binom{|\Sigma|}{k}\cdot |C|\right)^{O(1)}$;
    \item if $C$ has a satisfying assignment of weight at most $k$, then
    the same assignment satisfies $C'$;
    \item if every satisfying assignment of $C$ has weight at least $k+1$,
    then every satisfying assignment of $C'$ has weight more than $h$.
\end{itemize}
\end{theorem}

\begin{proof}
For every $i\in[m]$ and every $T\in\binom{\Sigma}{k}$, write
\[
    N_{i,T}:=f_i^{-1}(T).
\]
Define
\[
    C'(x):=
    \bigwedge_{i\in[m]}
    \bigvee_{T\in\binom{\Sigma}{k}}
    C\bigl(x\land \mathbf 1_{N_{i,T}}\bigr).
\]
This is a monotone circuit with two additional outer layers, hence an
MMSA$_{s+2}$ instance, and its size is bounded by
$\left(m\cdot \binom{|\Sigma|}{k}\cdot |C|\right)^{O(1)}$.

For completeness, let $x$ be a satisfying assignment of $C$ with support
$X\subseteq[n]$ and $|X|\le k$. For every $i\in[m]$, choose
$T_i\in\binom{\Sigma}{k}$ such that
\[
    f_i(X)\subseteq T_i.
\]
This is possible because $|f_i(X)|\le |X|\le k$ and $|\Sigma|\ge k$.
Then $X\subseteq N_{i,T_i}$, so
\[
    x\land \mathbf 1_{N_{i,T_i}}=x.
\]
Therefore every outer conjunction of $C'$ is satisfied, and hence $C'(x)=1$.

For soundness, let $x$ be any satisfying assignment of $C'$ and let
$X=\operatorname{supp}(x)$. Since $C'(x)=1$, for every $i\in[m]$ there exists
some $T_i\in\binom{\Sigma}{k}$ such that
\[
    C\bigl(x\land \mathbf 1_{N_{i,T_i}}\bigr)=1.
\]
By the soundness assumption on $C$, the assignment
$x\land \mathbf 1_{N_{i,T_i}}$ must have weight at least $k+1$. Equivalently,
\[
    |X\cap N_{i,T_i}|\ge k+1
    \qquad\text{for every } i\in[m].
\]

Suppose toward contradiction that $|X|\le h$. Since $\mathcal F$ is an
$(n,m,\Sigma,h)$-perfect hash family, there exists some $i^\star\in[m]$ such
that $f_{i^\star}$ is injective on $X$. For this $i^\star$, every
$T\in\binom{\Sigma}{k}$ can contain images of at most $k$ elements of $X$.
Thus
\[
    |X\cap f_{i^\star}^{-1}(T)|\le k
    \qquad\forall T\in\binom{\Sigma}{k},
\]
contradicting
\[
    |X\cap N_{i^\star,T_{i^\star}}|\ge k+1.
\]
Hence $|X|>h$.
\end{proof}

\begin{remark}
    We remark that the goal of Marx \cite{Marx13} was to rule out $f(k)$-factor FPT approximation algorithms for any computable function $f$. Thus, it was sufficient for him to choose $h$ to be larger than the approximation factor $f(k)$. In his hardness proof, $h$ was used as the parameter of an intermediate \emph{FPT cost-approximation algorithm}, which in turn forced $h$ to depend only on $k$. In contrast, our argument proceeds more directly, allowing us to choose $h$ to be as large as $n^{O(1/k)}$.
\end{remark}

We apply \Cref{thm:mmsa_reduction_phf} using an
$(n,m,\Sigma,h)$-perfect hash family from \Cref{lem:perfect_hash_family}
with $h=n^{\Theta(1/k)}, |\Sigma|=h^2$ and $m=h^{O(1)}\log n$. Then
\[
    \binom{|\Sigma|}{k}\le h^{2k},
\]
so the reduction size and running time are bounded by
\[
    \left(m\cdot \binom{|\Sigma|}{k}\cdot |C|\right)^{O(1)}
    =
    \left(h^{O(1)}\log n\cdot h^{2k}\cdot |C|\right)^{O(1)} \le n^{O(1)}.
\]
It was known \cite{CHKX06} that under $\hypoeth$, $k$-MMSA$_2$ (a.k.a. $k$-SetCover) cannot be solved in $n^{o(k)}$ time. We thus have the following immediate corollary:

\begin{corollary}
\label{cor:mmsafour_eth}
    Assuming $\hypoeth$, no $n^{o(k)}$-time algorithm can approximate \mmsafour{} within some $n^{O(1/k)}$ factor.
\end{corollary}

Using the same set of parameters and invoking the padding lemma (\Cref{lem:mmsa_padding}), we also have the following $\wtwo$-hardness:

\begin{corollary}
\label{cor:mmsafour_wtwo}
    Assuming $\wtwo\neq\fpt$, no FPT algorithm can approximate \mmsafour{} within any $n^{o(1)}$ factor.
\end{corollary}

\Cref{cor:mmsafour_eth} and \Cref{cor:mmsafour_wtwo} improve upon  \Cref{lem:klmkl_mmsa4} and \Cref{lem:marx13kmmsa4}, respectively.

\paragraph*{Acknowledgements.}
We thank Yican Sun for helpful early discussions with the authors on the $k$-MMSA problem.

\bibliographystyle{alpha}
\bibliography{ref}

\end{document}